\documentclass[journal]{IEEEtran}
\usepackage{amsmath,graphicx,hyperref}
\usepackage{amsmath, amsfonts, mathtools, amsthm, amssymb}
\usepackage{comment}
\usepackage[ruled,vlined,linesnumbered]{algorithm2e}
\usepackage{nomencl}
\usepackage{xcolor}
\usepackage{stfloats}
\usepackage{algpseudocode}
\usepackage{kotex}
\usepackage[caption=false,font=normalsize,labelfont=rm,textfont=rm]{subfig}
\usepackage{nomencl}
\usepackage{multirow}
\usepackage{upgreek}
\usepackage{bm}

\DeclareMathOperator*{\argmax}{arg\,max}

\newtheorem{lemma}{Lemma}

\begin{document}

\title{QoS-Constrained Resource Pattern Design for V2X-ISAC Systems}

\author{Hanyoung~Park,
        Gangmin Kim, 
        Yoo-Seung Song,
        and~Ji-Woong~Choi
\thanks{This work was supported by the Institute of Information \& Communications Technology Planning \& Evaluation (IITP) grant funded by the Korean government (MSIT) (No. RS-2025-02283230). \textit{(Corresponding author: Ji-Woong Choi.)}}
\thanks{H. Park was with the Department of Electrical Engineering and Computer Science, Daegu Gyeongbuk Institute of Science and Technology (DGIST), Daegu 42988, South Korea. He is now with Institute of Innovation for Future Army, Republic of Korea Army, Daejeon 34051, South Korea (e-mail: prkhnyng@dgist.ac.kr).}
\thanks{G. Kim is with the Interdisciplinary Studies of Artificial Intelligence, DGIST, Daegu 42988, South Korea (e-mail: kjkj0513@dgist.ac.kr).}
\thanks{Y.-S. Song is with
Autonomous Driving Intelligence Research Section, Electronics and Telecommunications Research Institute (ETRI), Daejeon 34129, South Korea (e-mail: yssong00@etri.re.kr).}
 \thanks{J.-W. Choi is with the Department
of Electrical Engineering and Computer Science, DGIST, Daegu 42988, South Korea (e-mail: jwchoi@dgist.ac.kr).}
}

\maketitle

\begin{abstract}
Integrated sensing and communication (ISAC) has emerged as a promising approach for vehicle-to-everything (V2X) systems by enabling communication and sensing over shared radio resources without additional installation of dedicated sensors. However, candidate resources may experience different communication qualities due to varying channel conditions and resource contention, which should be considered when designing sensing resource patterns. In this letter, we propose a quality-of-service (QoS)-constrained resource pattern selection framework that jointly considers communication reliability and sensing performance. The sensing objective is formulated based on the Fisher information matrix for joint range and velocity estimation, while a minimum packet reception ratio (PRR) is imposed as the communication QoS constraint. To avoid the complexity of exhaustive search, a greedy selection algorithm with one-swap refinement is developed. Simulation results show that the proposed method improves sensing performance over conventional resource patterns while satisfying the PRR requirement and achieves performance close to the exhaustive-search optimum and sensing-only optimum.
\end{abstract}


\IEEEpeerreviewmaketitle

\section{Introduction}
\label{sec:intro}

Vehicle-to-everything (V2X) communication has been widely regarded as a key enabling technology for future intelligent transportation systems (ITSs), as it enables information exchange among vehicles, infrastructure, networks, and pedestrians~\cite{ref:v2x1, ref:v2x2}. In addition to enabling rapid data exchange to provide safe and reliable transportation and support advanced driver assistance systems (ADASs), V2X communication can also facilitate the transmission and sharing of sensing information in a real-time manner~\cite{ref:v2x_energy,ref:3gpp22186}. In this context, integrated sensing and communication (ISAC) enables vehicles to acquire sensing information without requiring dedicated radar sensors, thereby avoiding the additional hardware cost associated with separate sensing systems~\cite{ref:v2x_isac, ref:v2x_isac2}. 

Owing to these advantages, the integration of ISAC into connected and autonomous vehicles has recently attracted considerable research interest.
To support the joint communication and sensing service with better performance, previous studies on V2X-ISAC have primarily focused on beamforming~\cite{ref:rw_beamforming}, sensing latency~\cite{ref:rw_sensinglatency}, power allocation~\cite{ref:rw_subcarrier}, localization~\cite{ref:rw_localization}, waveform design~\cite{ref:rw_waveform}, among others. However, the performance of these techniques inherently depends on the underlying resource pattern used for sensing and communication~\cite{ref:rw_v2xsidelink}. Despite its importance, resource pattern design has received relatively limited attention. Although recent studies considered resource pattern design for V2X-ISAC, it did not deeply account for the resource placement considering time- and frequency-varying channel quality across V2X radio slots~\cite{ref:rw_pattern, ref:rw_pattern2, ref:rw_pattern3}. To address this limitation, this work proposes a resource pattern placement framework for V2X systems that explicitly considers slot-dependent channel quality while jointly considering communication quality-of-service (QoS) and sensing performance.
The main contributions of this work are summarized as follows:
\begin{itemize}
    \item We develop an ISAC resource pattern design framework that explicitly accounts for slot-dependent channel quality.
    \item The proposed method constructs the resource pattern by selecting resources that improve sensing performance while satisfying the communication QoS constraint based on the Fisher information matrix.
    \item Compared with conventional resource pattern baselines, the proposed method improves sensing performance while maintaining communication performance. In particular, it achieves lower Cramér–Rao lower bounds (CRLBs) for range and velocity estimation at the same packet reception ratio (PRR).
    \item To avoid prohibitive complexity of exhaustive search, a resource selection algorithm is developed based on greedy selection followed by local swap refinement. The proposed method achieves performance close to that of the exhaustive-search optimum, while substantially reducing the computational complexity.
\end{itemize}

The rest of this paper is organized as follows. In Section~\ref{sec:prerequisites}, the prerequisites for the proposed method are described. In Section~\ref{sec:proposed}, the proposed method for resource pattern design is described in detail. In Section~\ref{sec:results}, the numerical results of the simulation and the performance analysis
based on them are described. Then, Section~\ref{sec:conclusion} concludes the paper.

\subsubsection*{Notation}
Uppercase and lowercase boldface letters denote matrices and vectors, respectively. For an arbitrary matrix $\mathbf{A}$, its $i$-th column vector is denoted by $\mathbf{a}_i$, and its $j$-th element is denoted by $A_{ij}$. For an arbitrary vector $\mathbf{b}$, its $k$-th element is denoted by $b_k$. $(\cdot)^T$ means transpose of a vector or a matrix.
$\det(\cdot)$ means the determinant of the corresponding matrix. Calligraphic letters denote sets. $|\mathcal{A}|$ is the cardinality of an arbitrary set $\mathcal{A}$.

\section{Prerequisites}\label{sec:prerequisites}

\subsection{System Model}
We consider a 5G-V2X sidelink Mode 2 system~\cite{ref:v2xmode2, ref:3gpp38214} in which a user equipment (UE) autonomously selects transmission resources based on its local sensing information.
The available sidelink resource pool consists of $N_t$ time slots and $N_f$ frequency subchannels.
The corresponding set of time-frequency resources is defined as
\begin{equation}
    \mathcal{R} = \{ (n, m): n\in\{1,...,N_t\}, m\in\{ 1,...,N_f \} \},
\end{equation}
where $n$ and $m$ denote the slot and subchannel indices, respectively.
Following the Mode 2 sensing procedure, resources that are regarded as occupied or reserved are excluded from the selection window~\cite{ref:3gpp38321, ref:3gpp38331}.
The remaining resources constitute the candidate set of the considered UE, denoted by $\mathcal{C}\subseteq \mathcal{R}$.
From $\mathcal{C}$, the considered UE selects $K$ resources for joint communication and sensing. The resulting resource pattern is represented by
\begin{equation}
    \mathcal{S} = \{r_1,...,r_K\}, r_k=(n_k,m_k),
\end{equation}
where $n_k$ and $m_k$ denote the slot and subchannel associated with the $k$-th selected resource.
Hence, a selected resource pattern satisfies $\mathcal{S}\subseteq \mathcal{C}$ and $|\mathcal{S}|=K$.
Let $F_k$ and $T_k$ denote the center frequency and transmission time associated with resource $r_k$, respectively. The set $\{ (F_k, T_k)  \}^K_{k=1}$ therefore determines the time-frequency geometry of the sensing waveform.

Meanwhile, due to autonomous resource selection in Mode 2, different candidate resources may experience different levels of contention from neighboring UEs. 
Assuming that the resource-selection events of neighboring UEs are conditionally independent, the collision probability associated with resource $r$ is given by
\begin{equation}
    P_{\mathrm{col}} (r) = 1- \prod_{j=1}^{N_\mathrm{neighbor}} (1-\pi_{j,r}),
\end{equation}
where $N_\mathrm{neighbor}$ denotes the number of neighboring UEs and $\pi_{j,r}$ means the resource-selection probability of neighboring UE $j$ for resource $r$.
Let $\bar{\gamma}_r$ denote the average received signal-to-interference-plus-noise ratio (SINR) over resource $r$. The probability of successful physical layer decoding with Rayleigh fading can be expressed as~\cite{ref:outage}
\begin{equation}
    P_{PHY}(r) = \exp\left( - \frac{\gamma_{\mathrm{th}}}{\bar{\gamma}_r} \right),
\end{equation}
where $\gamma_\mathrm{th}$ denotes the decoding threshold.
Accordingly, the packet success probability over resource $r$ is modeled as
\begin{equation}
    P_\mathrm{suc}(r) = [1-P_\mathrm{col}(r)]P_{PHY}(r).
\end{equation}
For a selected resource pattern $\mathcal{S}$, the average packet reception ratio (PRR) can be presented as
\begin{equation}
    \mathrm{PRR}(\mathcal{S}) = \frac{1}{K} \sum_{r\in\mathcal{S}} P_\mathrm{suc} (r).
\end{equation}

Meanwhile, each selected resource is assumed to occupy $Q$ subcarriers.
The frequency of the $q$-th subcarrier within resource $r_k$ is represented as
\begin{equation}
    f_{k,q} = F_k + \delta_q,
\end{equation}
where $\delta_q$ denotes the frequency offset from the resource center and satisfies $\sum_{q=1}^Q \delta_q = 0$.
For a sensing target characterized by range $R$ and radial velocity $v$, the corresponding round-trip delay and Doppler frequency are given by~\cite{ref:delaydoppler}
\begin{equation}
    \tau = \frac{2R}{c}, \quad \nu = \frac{2f_c v}{c},
\end{equation}
where $c$ is the speed of light and $f_c$ denotes the carrier frequency.

Also, after removing the known transmitted symbol, the sensing observation at the $q$-th subcarrier of the $k$-th selected resource is modeled as~\cite{ref:signal}
\begin{equation}
    y_{k,q} = \alpha e^{-j2\pi f_{k,q}\tau}e^{j2\pi T_k\nu} + w_{k,q},
\end{equation}
where $\alpha\in\mathbb{C}$ denotes the unknown complex target reflection coefficient and $w_{k,q}\sim \mathcal{CN}(0, \sigma^2)$ is additive complex Gaussian noise with variance $\sigma^2$.
The corresponding mean signal is therefore~\cite{ref:signaljsac}
\begin{equation}\label{eqn:meansignal}
    \mu_{k,q} = \alpha e^{-j2\pi f_{k,q}\tau} e^{j2\pi T_k\nu}.
\end{equation}

\subsection{CRLB Derivation}
Let the unknown parameter vector be defined as
\begin{equation}
    \bm{\uptheta} = [\alpha_R, \alpha_I, \tau, \nu]^T,
\end{equation}
where $\alpha_R = \mathrm{Re}(\alpha)$, $\alpha_I = \mathrm{Im}(\alpha)$ are the real and imaginary parts of $\alpha$, respectively. 
Based on this formulation, the $(i, j)$-th entry of the Fisher information matrix $\mathbf{J}$~\cite{ref:fisher} is then given by
\begin{equation}
    J_{i,j} = \frac{2}{\sigma^2} \mathrm{Re}\left\{ \sum_{k=1}^K \sum_{q=1}^Q \frac{\partial \mu_{k,q}^*}{\partial \theta_i}\frac{\partial \mu_{k,q}}{\partial \theta_j} \right\}.
\end{equation}
The derivatives with respect to delay and Doppler are
\begin{equation}
 \frac{\partial \mu_{k,q}}{\partial \tau} = -j2\pi f_{k,q} \mu_{k,q}, 
    \quad\frac{\partial \mu_{k,q}}{\partial \nu}
    = j2\pi T_k \mu_{k,q}.
\end{equation}
Also, the derivatives with respect to the real and imaginary parts of the target coefficient are
\begin{equation}
    \frac{\partial \mu_{k,q}}{\partial \alpha_R} = e^{-j2\pi f_{k,q}\tau }e^{j2\pi T_k\nu},\quad
    \frac{\partial \mu_{k,q}}{\partial \alpha_I} = je^{-j2\pi f_{k,q}\tau }e^{j2\pi T_k\nu}.
\end{equation}
Since the complex target coefficient $\alpha$ is an unknown nuisance parameter, we partition the Fisher information matrix as
\begin{equation}
    \mathbf{J} = \begin{bmatrix}
        \mathbf{J}_{\alpha\alpha} & \mathbf{J}_{\alpha x} \\
        \mathbf{J}_{x\alpha} & \mathbf{J}_{xx} 
    \end{bmatrix}
    =
    \begin{bmatrix}
        J_{\alpha_R\alpha_R} & J_{\alpha_R\alpha_I} & J_{\alpha_R\tau} & J_{\alpha_R\nu} \\
        
        {J}_{\alpha_I\alpha_R} & {J}_{\alpha_I\alpha_I} & {J}_{\alpha_I\tau} & {J}_{\alpha_I\nu} \\

        J_{\tau\alpha_R} & J_{\tau\alpha_I} & J_{\tau\tau} & J_{\tau\nu} \\

        J_{\nu\alpha_R} & J_{\nu\alpha_I} & J_{\nu\tau} & J_{\nu\nu} \\
    \end{bmatrix}
    ,
\end{equation}
where $\mathbf J_{\alpha\alpha}$ contains the information associated with the nuisance parameters $[\alpha_R,\alpha_I]^T$, $\mathbf J_{xx}$ contains the information associated with the parameters of interest $\mathbf{x}=[\tau,\nu]^T$, and $\mathbf J_{\alpha x}=\mathbf J_{x\alpha}^{T}$ represents the cross-information between the two parameter groups.
For example, since the exponential terms in \eqref{eqn:meansignal} have unit magnitude,
\begin{equation}
    J_{\alpha_R \alpha_R} = J_{\alpha_I \alpha_I} = \frac{2KQ}{\sigma^2}, \quad J_{\alpha_R \alpha_I} = 0.
\end{equation}
Similarly, the delay-Doppler submatrix $\mathbf{J}_{xx}$ is obtained as
\begin{equation}
    \mathbf{J}_{xx} = \beta \begin{bmatrix}
        \sum_{k,q} f^2_{k,q} & -\sum_{k,q} f_{k,q} T_k \\ 
        -\sum_{k,q} f_{k,q} T_k & \sum_{k,q} T_k^2,
    \end{bmatrix}
\end{equation}
where $\beta = 8\pi^2|\alpha|^2/\sigma^2$.
The equivalent Fisher information matrix for delay and Doppler is obtained through the Schur complement as
\begin{equation}\label{eqn:schur}
    \mathbf{J}_\mathrm{eq} = \mathbf{J}_{xx}  - \mathbf{J}_{x\alpha}\mathbf{J}_{\alpha\alpha}^{-1}\mathbf{J}_{\alpha x}.   
\end{equation}
Lemma~\ref{lem:1} shows this equivalent Fisher information matrix can be simplified with frequency and time locations.
\begin{lemma}\label{lem:1}
    Let $\bar{f}$ and $\bar{t}$ denote the average frequency and time location over all $KQ$ sensing observations, respectively. Also, let the $KQ$ frequency and time locations be indexed by $\ell$. Then, the equivalent Fisher information matrix can be simplified as
    \begin{equation}
        \mathbf{J}_\mathrm{eq} = \beta
        \begin{bmatrix}
            S_{ff} & -S_{ft} \\
            -S_{ft} & S_{tt}
        \end{bmatrix},
    \end{equation}
    where
    \begin{equation}
        S_{ff} = \sum_{\ell=1}^{KQ} (f_\ell - \bar{f})^2,
        S_{tt} = \sum_{\ell=1}^{KQ} (t_\ell - \bar{t})^2,
        S_{ft} = \sum_{\ell=1}^{KQ} (f_\ell - \bar{f})(t_\ell - \bar{t}).
    \end{equation}
\end{lemma}
\begin{proof}
    Please refer to Appendix~\ref{sec:app1}.
\end{proof}
Therefore, the determinant of the equivalent Fisher information matrix becomes
\begin{equation}
    \det(\mathbf{J}_\mathrm{eq}) = \beta^2 D,
\end{equation}
where $D=S_{ff}S_{tt}-S_{ft}^2$.
Accordingly, the CRLBs for delay and Doppler estimation are given by
\begin{equation}
    \mathrm{CRLB}_\tau = \frac{S_{tt}}{\beta D}, \quad \mathrm{CRLB}_\nu = \frac{S_{ff}}{\beta D}.
\end{equation}
Using the relationships between $(\tau,\nu)$ and $(R, v)$, the corresponding range and velocity CRLBs are 
\begin{equation}
    \mathrm{CRLB}_R = \frac{c^2}{4} \frac{S_{tt}}{\beta D}, \quad \mathrm{CRLB}_v = \frac{c^2}{4f_c^2}\frac{S_{ff}}{\beta D}.
\end{equation}

\section{QoS-Constrained Resource Pattern Selection}\label{sec:proposed}
Since $
    \det(\mathbf{J}_\mathrm{eq}) = \beta^2 D(\mathcal{S}),$
and $\beta$ is independent of the selected resource pattern, maximizing the determinant of the equivalent Fisher information matrix is equivalent to maximizing $D(\mathcal{S})$.
Accordingly, we adopt the $D$-optimal design criterion. 
The $D$-optimal criterion is particularly suitable for the considered joint range-velocity estimation problem because maximizing $\det(\mathbf{J}_\mathrm{eq})$ minimizes the volume of the corresponding joint CRLB uncertainty region. 
Based on these derivations, the optimization problem for resource pattern design can be formulated as
\begin{equation}\label{eqn:problem}
\begin{aligned}
    & \mathcal{S}^\star = \argmax_\mathcal{S}  D(\mathcal{S}) \\
    & \textnormal{s.t. } \mathrm{PRR}(\mathcal{S}) \geq \rho_\mathrm{min}, \\
    & \quad \ \ |\mathcal{S}| = K, \\
    & \quad \ \ \mathcal{S} \subseteq \mathcal{C}, 
\end{aligned}
\end{equation}
where $\rho_\mathrm{min}$ denotes the minimum required PRR.
The proposed formulation therefore selects, among the communication-feasible resource patterns, the one that maximizes the joint delay-Doppler information.

To solve the formulated problem in \eqref{eqn:problem}, exhaustive search can be employed; however, it may incur a prohibitively high computational burden. Therefore, we propose a suboptimal approach that achieves performance comparable to that of exhaustive search with substantially reduced computational complexity. 
The steps of the proposed resource selection method are organized in Algorithm~\ref{alg:algo}.
The algorithm is initialized by selecting a pair of resources from different slots. Among all candidate pairs that can still be extended to satisfy the final PRR requirement, the pair yielding the largest value of $D(\mathcal{S})$ is selected.
Starting from this pair, resources are added sequentially. At an intermediate stage with a partial set $\mathcal{S}_s$, a candidate resource $r$ is considered only if the remaining unselected resources can still provide sufficient communication reliability to satisfy 
\begin{equation}
    \sum_{r'\in \mathcal{S}_s \cup \{r\}} P_\mathrm{suc}(r') + \sum_{\ell=1}^{K-|\mathcal{S}_s|-1} p_{\ell}^\mathrm{max} \geq K\rho_\mathrm{min},
\end{equation}
where $p_{\ell}^\mathrm{max}$ denotes the largest available packet-success probability in the remaining admissible slots. Among the candidates satisfying this condition, the resource that maximizes $D(\mathcal{S}_s \cup \{r\})$ is selected. This procedure is repeated until $K$ resources are obtained.

To compensate for the myopic nature of greedy selection, the resulting pattern is further refined by a 1-swap search.
Specifically, one selected resource is replaced by one unselected candidate whenever the new pattern satisfies the PRR and slot constraints and yields a larger value of $D(\mathcal{S})$.
The swap operation is repeated until no further improvement is possible.
Then, the exhaustive benchmark evaluates up to $\begin{pmatrix}
    |\mathcal{C}| \\ K
\end{pmatrix}$ resource patterns, whereas the proposed method requires approximately $\mathcal{O}(|\mathcal{C}|^2)$ operations for pair initialization, $\mathcal{O}(K|\mathcal{C}|)$ for greedy selection, and $\mathcal{O}(IK|\mathcal{C}|)$ for the 1-swap refinement, where $I$ is the number of refinement iterations. Therefore, the proposed method substantially reduces the search complexity. 

\begin{algorithm}[t]
\caption{Greedy Resource Selection with 1-Swap Refinement}\label{alg:algo}
\SetAlgoLined
\textbf{input:} Candidate resource set $\mathcal{C}$, number of selected resources $K$, minimum PRR requirement $\rho_\mathrm{min}$

\textbf{initialize} $\mathcal{S}\leftarrow \emptyset$.

Select the admissible resource pair with the largest $D(\mathcal{S})$ that can be extended to satisfy the PRR requirement.

\While{$|\mathcal{S}|<K$}{
    For each admissible $r\in\mathcal{C}\backslash\mathcal{S}$, check whether $\mathcal{S}\cup \{r\}$ can still be completed to satisfy $\mathrm{PRR}\geq \rho_\mathrm{min}$.

    Select $r^\star = \argmax_r D(\mathcal{S}\cup \{r\})$.

    Update $\mathcal{S}\leftarrow \mathcal{S} \cup \{r^\star\}$.
}
\Repeat{no improving feasible swap exists}{Search all feasible one-resource swaps.

Replace the current pattern with the feasible swap yielding the largest increase in $D(\mathcal{S})$.}

\Return{$\mathcal{S}$}    
\end{algorithm}

\section{Simulation Results}\label{sec:results}
\subsection{Simulation Configuration}
We consider a 5G-V2X sidelink ISAC system operating at a carrier frequency of 5.9 GHz with a subcarrier spacing of 30 kHz. The resource pool consists of 8 time slots and 8 subchannels, where each slot has a duration of 0.5 ms and each subchannel contains 10 physical resource blocks (PRBs), corresponding to 120 subcarriers. Among the available time-frequency resources, $K=4$ resources are selected for ISAC operation, with at most one resource selected per slot.
Unless otherwise stated, 10 neighboring UEs are considered. The SINR of each candidate resource is independently drawn with a mean of 10 dB and a standard deviation of 1.5 dB, and the SINR threshold for PHY decoding is set to 0 dB. The minimum packet reception ratio (PRR) requirement is set to $\rho_{\min}=0.9$.

The proposed QoS-aware greedy selection with one-swap refinement, ``Greedy + Swap'', is compared with four baselines, ``Sensing-only'', ``Exhaustive search'', ``Fixed equispaced'', and ``Random feasible''. 
The ``Sensing-only'' selects the resource pattern solely based on sensing performance without considering the communication QoS constraint. The ``Exhaustive search'', which is the exhaustive-search-based QoS-constrained optimum, selects the best pattern among all patterns satisfying the PRR requirement and is used as an optimal benchmark. The ``Fixed equispaced'' selects resources with approximately uniform spacing in the time-frequency domain, while the ``Random feasible'' randomly selects one pattern from the set of patterns satisfying the PRR constraint.
The simulation results are averaged over 10000 Monte Carlo realizations. 

\subsection{Numerical Results and Analysis}


Fig.~\ref{fig:CRLB} shows the range and velocity CRLBs depending on the minimum required PRR and the number of neighboring UEs. As shown in Figs.~\ref{subfig:rangePRR} and \ref{subfig:velocityPRR}, the CRLBs generally increase slightly as the PRR constraint becomes more stringent. Nevertheless, the proposed ``Greedy + Swap'' method closely matches the ``Exhaustive search'' benchmark and maintains performance comparable to the ``Sensing-only'' scheme up to a minimum PRR of 0.9.
Figs.~\ref{subfig:rangeN} and \ref{subfig:velocityN} show similar trends with respect to the number of neighboring UEs. The proposed method again achieves CRLBs close to both the exhaustive-search benchmark and the ``Sensing-only'' scheme. 

Fig.~\ref{fig:time} shows the normalized runtime versus the required minimum PRR and the number of neighboring UEs. While the runtime of the ``Greedy + Swap'' method changes only slightly, that of the ``Exhaustive search'' method decreases as the constraints become more stringent due to the reduced number of candidate patterns. Nevertheless, ``Greedy + Swap'' requires less than half the runtime of ``Exhaustive search'' while achieving comparable performance even under more stringent pattern constraints.

\begin{figure}
    \centering
    \subfloat[\label{subfig:rangePRR}]{\includegraphics[width=0.5\columnwidth]{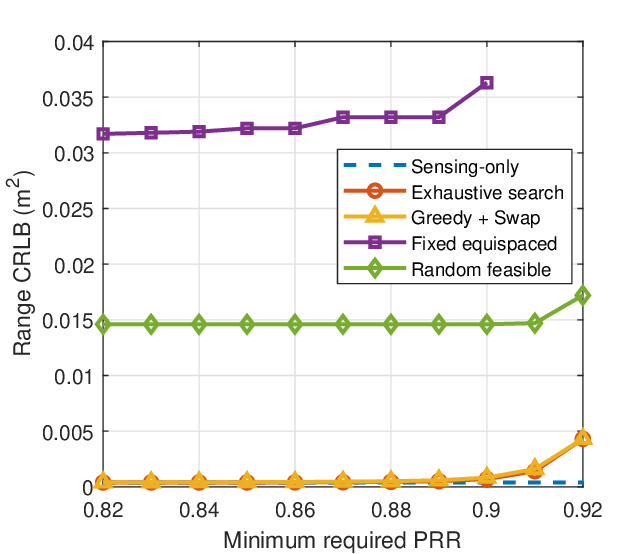}}
    \subfloat[\label{subfig:velocityPRR}]{\includegraphics[width=0.5\columnwidth]{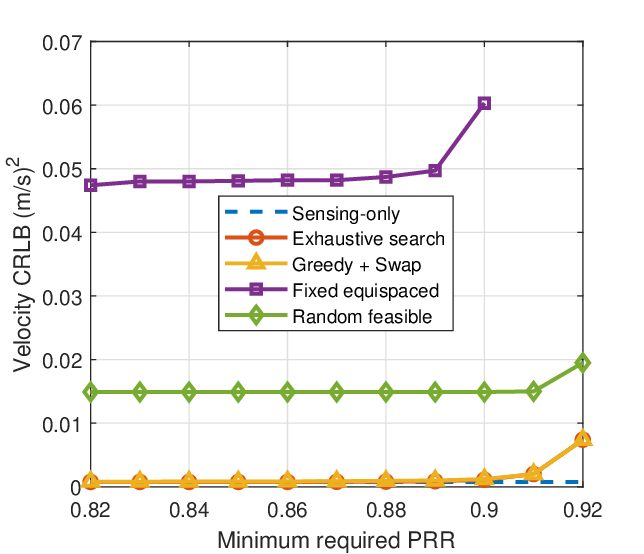}}

    \subfloat[\label{subfig:rangeN}]{\includegraphics[width=0.5\columnwidth]{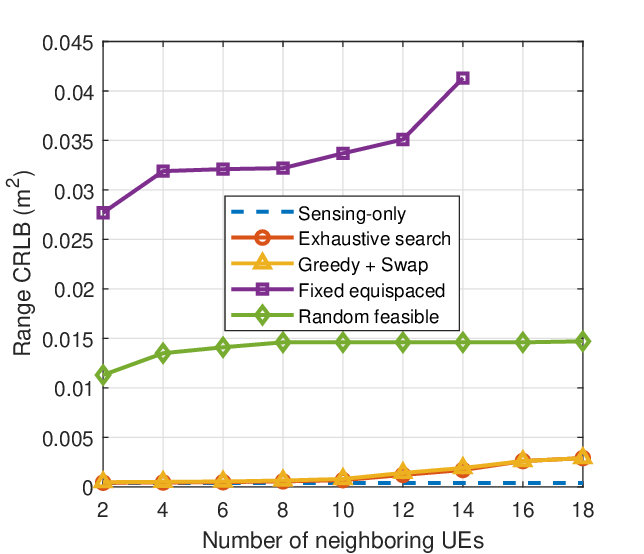}}
    \subfloat[\label{subfig:velocityN}]{\includegraphics[width=0.5\columnwidth]{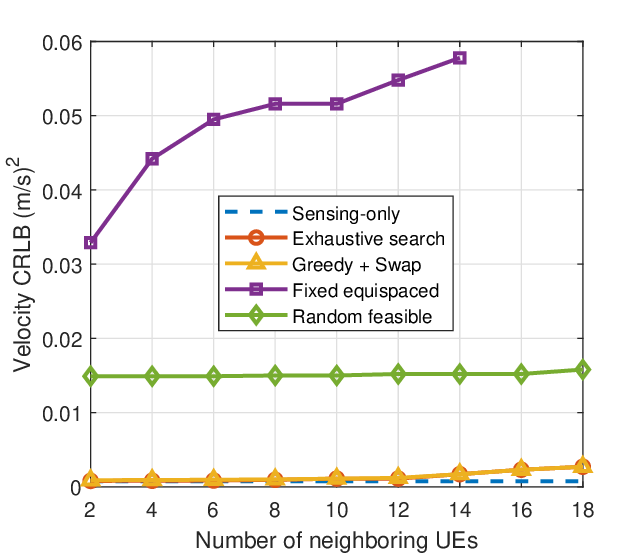}}
    
    \caption{CRLB depending on the required minimum PRR and the number of neighboring UEs. CRLBs of (a) range and (b) velocity depending on the required minimum PRR, (c) range and (d) velocity depending on the number of neighboring UEs.}
    \label{fig:CRLB}
\end{figure}

    

    

\begin{figure}
    \centering
    \subfloat[Required minimum PRR]{\includegraphics[width=0.5\columnwidth]{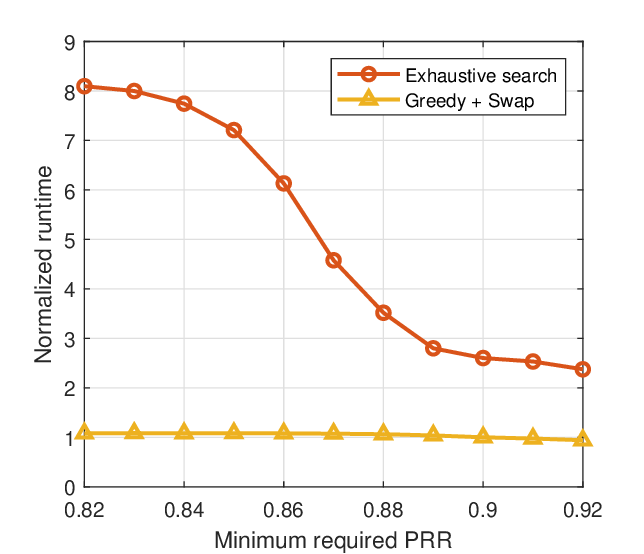}}
    \subfloat[Number of neighboring UEs]{\includegraphics[width=0.5\columnwidth]{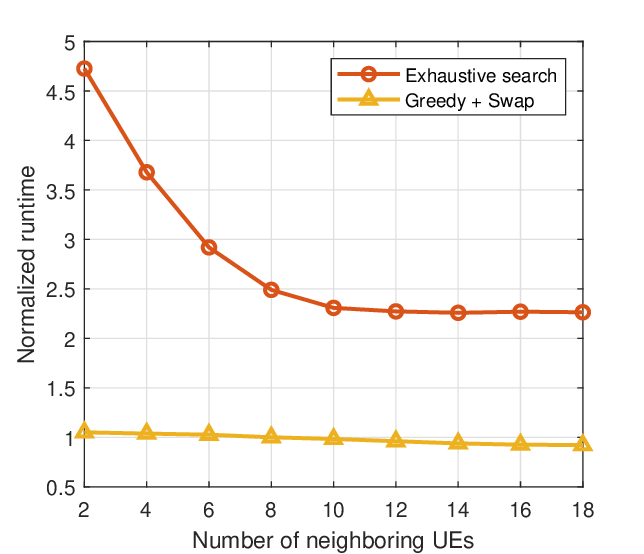}}
    
    \caption{Normalized runtime depending on the required minimum PRR and the number of neighboring UEs. The runtime is normalized to the case with 10 UEs and a minimum required PRR of 0.9.}
    \label{fig:time}
\end{figure}

\section{Conclusion}\label{sec:conclusion}
In this paper, we proposed a QoS-constrained resource pattern selection method for V2X-ISAC systems. By deriving the Fisher information matrix for joint range and velocity estimation, the resource pattern design was formulated to maximize sensing information while satisfying a minimum PRR requirement. To reduce the computational complexity of exhaustive search, a greedy resource selection algorithm followed by one-swap refinement was developed. Simulation results showed that the proposed method achieves range and velocity CRLB performance close to that of the exhaustive-search optimum while satisfying the communication QoS requirement. Moreover, the proposed method maintains sensing performance comparable to the sensing-only scheme.

\appendices
\section{Proof of Lemma \ref{lem:1}}\label{sec:app1}
Evaluating the Schur complement in \eqref{eqn:schur} gives
\begin{equation}
\begin{aligned}
    [\mathbf{J}_\mathrm{eq}]_{11} &= \beta\left[\sum_{\ell=1}^{KQ} f_\ell^2 - \frac{1}{KQ}\left( \displaystyle\sum_{\ell=1}^{KQ} f_\ell \right)^2\right] \\
    [\mathbf{J}_\mathrm{eq}]_{12} &= -\beta \left[ \displaystyle\sum_{\ell=1}^{KQ} f_\ell t_\ell - \frac{1}{KQ} \left( \displaystyle\sum_{\ell=1}^{KQ} f_\ell \right) \left( \displaystyle\sum_{\ell=1}^{KQ} t_\ell \right) \right] \\
    & = [\mathbf{J}_\mathrm{eq}]_{21},\\
    [\mathbf{J}_\mathrm{eq}]_{22} &= \beta \left[\displaystyle\sum_{\ell=1}^{KQ} t_\ell^2 - \frac{1}{KQ}\left( \displaystyle\sum_{\ell=1}^{KQ} t_\ell \right)^2\right]. 
\end{aligned}
\end{equation}
Using $\bar{f} = \frac{1}{KQ}\sum_{\ell=1}^{KQ} f_\ell$, $\bar{t} = \frac{1}{KQ}\sum_{\ell=1}^{KQ} t_\ell$,
and
\begin{equation}
\begin{aligned}
 \sum_{\ell=1}^{KQ}(f_\ell - \bar{f})^2 
    & = \sum_{\ell=1}^{KQ} f_\ell^2 -2\bar{f}\sum_{\ell=1}^{KQ} f_\ell  + \sum_{\ell=1}^{KQ} \bar{f}^2 
    \\& = \sum_{\ell=1}^{KQ} f_\ell^2 - 2KQ\bar{f}^2 + KQ \bar{f}^2
    \\& = \sum_{\ell=1}^{KQ} f_\ell^2 - KQ\bar{f}^2
    \\& = \sum_{\ell=1}^{KQ} f_\ell^2 - \frac{1}{KQ}\left( \sum_{\ell=1}^{KQ} f_\ell \right)^2,
\end{aligned}
\end{equation}
since $\sum_{\ell=1}^{KQ} f_\ell = KQ\bar{f}$ and
\begin{equation}
    KQ\bar{f}^2 = KQ\left( \frac{1}{KQ} \sum_{\ell=1}^{KQ} f_\ell \right)^2 = 
    \frac{1}{KQ} \left(  \sum_{\ell=1}^{KQ} f_\ell \right)^2,
\end{equation}
and same approach for $t$, the diagonal terms can be simplified with $S_{ff}$ and $S_{tt}$. Also, cross terms can be simplified with $-S_{ft}$ using the fact that
\begin{equation}
    \sum_{\ell=1}^{KQ} (f_\ell - \bar{f})(t_\ell - \bar{t}) = \sum_{\ell=1}^{KQ} f_\ell t_\ell - \frac{1}{KQ}\left( \sum_{\ell=1}^{KQ} f_\ell \right)\left( \sum_{\ell=1}^{KQ} t_\ell\right).
\end{equation}

\bibliographystyle{IEEEtran}
\bibliography{bibtex/bib/IEEEabrv,refs}

\end{document}